\documentclass[submission,copyright,creativecommons]{eptcs}
\providecommand{\event}{QPL 2020} 
\usepackage{breakurl}             
\usepackage{underscore}           

\DeclareMathAlphabet{\mathcal}{OMS}{cmsy}{m}{n}
\usepackage{amssymb}
\usepackage{amsthm} 
\usepackage{amsmath}
\usepackage{amsfonts}
\usepackage{tikz-cd}
\usepackage{bussproofs}
\usepackage{mathrsfs}
\usepackage{mathtools}
\usepackage{array}
\usepackage{extarrows}
\usepackage[b]{esvect}
\tikzset{%
    symbol/.style={%
        draw=none,
        every to/.append style={%
            edge node={node [sloped, allow upside down, auto=false]{$#1$}}}
    }
}
\makeatletter
\newsavebox{\@brx}
\newcommand{\llangle}[1][]{\savebox{\@brx}{\(\m@th{#1\langle}\)}%
\mathopen{\copy\@brx\kern-0.5\wd\@brx\usebox{\@brx}}}
\newcommand{\rrangle}[1][]{\savebox{\@brx}{\(\m@th{#1\rangle}\)}%
\mathclose{\copy\@brx\kern-0.5\wd\@brx\usebox{\@brx}}}
\makeatother

\newtheorem{theorem}{Theorem}[section]
\newtheorem{lemma}[theorem]{Lemma}
\newtheorem{proposition}[theorem]{Proposition}
\newtheorem{corollary}[theorem]{Corollary}

\newtheorem{definition}[theorem]{Definition}

\theoremstyle{remark}

\newtheorem{example}[theorem]{Example}

\renewcommand{\H}{\mathcal H}

\newcommand{\Q}{\mathcal Q}

\newcommand{\V}{\mathcal V}
\newcommand{\W}{\mathcal W}
\newcommand{\X}{\mathcal X}
\newcommand{\Y}{\mathcal Y}
\newcommand{\Z}{\mathcal Z}

\newcommand{\CC}{\mathbb C}

\newcommand{\NN}{\mathbb N}

\newcommand{\RR}{\mathbb R}

\newcommand{\fF}{\mathfrak F}

\renewcommand{\:}{\colon}
\newcommand{\To}{\rightarrow}

\newcommand{\subsetof}{\subseteq}

\newcommand{\suchthat}{\,|\,}

\newcommand{\tensor}{\otimes}
\newcommand{\iso}{\cong}

\newcommand{\union}{\cup}

\newcommand{\Tr}{\mathrm{Tr}}
\newcommand{\At}{\mathrm{At}}

\newcommand{\Eval}{\mathrm{Eval}}
\newcommand{\fix}{\mathrm{fix}}

\newcommand{\Alg}{\mathbf{OpAlg}}
\newcommand{\Set}{\mathbf{Set}}
\newcommand{\POS}{\mathbf{POS}}
\newcommand{\CPO}{\mathbf{CPO}}
\newcommand{\qRel}{\mathbf{qRel}}
\newcommand{\qSet}{\mathbf{qSet}}
\newcommand{\qPOS}{\mathbf{qPOS}}

\newcommand{\qCPO}{\mathbf{qCPO}}

\newcommand{\Rel}{\mathbf{Rel}}

\DeclareMathAlphabet{\mathpzc}{OT1}{pzc}{m}{it}

\title{Quantum CPOs}
\author{Andre Kornell
\institute{Tulane University\\ Louisiana, USA}
\email{akornell@tulane.edu}
\and
Bert Lindenhovius 
\institute{Tulane University\\
Louisiana, USA}
\email{alindenh@tulane.edu}
\and
Michael Mislove\institute{Tulane University\\
Louisiana, USA}
\email{mislove@tulane.edu}
}

\def\titlerunning{qCPOs}
\def\authorrunning{A. Kornell, B. Lindenhovius \& M. Mislove}

\allowdisplaybreaks

\begin{document}

\maketitle

\begin{abstract}
We introduce the monoidal closed category $\qCPO$ of \emph{quantum cpos}, whose objects are `quantized' analogs of $\omega$-complete partial orders (cpos). The category $\qCPO$ is enriched over the category $\CPO$ of cpos, and contains both $\CPO$, and the opposite of the category $\mathbf{FdAlg}$ of finite-dimensional von Neumann algebras as monoidal subcategories. 
We use $\qCPO$ to construct a sound model for the quantum programming language Proto-Quipper-M (PQM) extended with term recursion, as well as a sound and computationally adequate model for the Linear/Non-Linear Fixpoint Calculus (LNL-FPC), which is both an extension of the Fixpoint Calculus (FPC) with linear types, and an extension of a circuit-free fragment of PQM that includes recursive types.
\end{abstract}

\section{Introduction}\label{sec:intro}

This is a progress report on an ongoing project to develop semantic models for quantum programming languages, with a particular focus supporting term recursion and type recursion. We introduce the category $\qCPO$ of `quantum cpos' and Scott continuous maps, which can be regarded as a `quantized' version of $\CPO$, the category of $\omega$-complete partial orders (cpos) and Scott continuous maps. \textbf{CPO} is a crucial building block for denotational models of classical programming languages with recursive types~\cite{lehman-smyth}; similarly, we show denotational models for quantum programming languages with recursive types can be built using $\qCPO$: we prove $\qCPO$ supports a sound and computationally adequate model of  Proto-Quipper-M (PQM) \cite{pqm-small} extended with term recursion, as well as sound model of the Linear/Non-Linear Fixpoint Calculus (LNL-FPC)~\cite{LMZ20} that is computationally adequate at non-linear types.

The models of high-level functional quantum programming languages such as PQM are based on \emph{linear / nonlinear models}, which arise from the work of Benton~\cite{benton-small,benton-wadler} on models of the linear lambda calculus. More precisely, a linear/non-linear model consists of a non-linear Cartesian closed category $\mathbf C$, a symmetric monoidal closed linear category $\mathbf L$, and a strong monoidal functor $F\colon \mathbf C\to\mathbf L$ that has a right adjoint:
\begin{equation}\label{diag:LNL}
\begin{tikzcd}
\mathbf{C}\ar[r,bend left,"F",""{name=A, below}] & \mathbf{L}\ar[l,bend left,"G",""{name=B,above}] \ar[from=A, to=B, symbol=\dashv]
\end{tikzcd}
\end{equation}

A result in~\cite{LMZ18} shows that the semantic models of PQM all are suitable linear/non-linear models, and results in~\cite{LMZ18} and \cite{LMZ20} show how to add term recursion and type recursion, respectively, to such models. 

Abstract LNL models are suitable for a wide range of circuit description languages beyond those aimed at quantum devices. Examples include nonlinear languages that also support linear types, such as recent work on languages supporting session types for concurrency~\cite{Pfenning}.
There also have been efforts to construct more concrete models for quantum programming languages based on a linear category closer to quantum computation. The category $\mathbf{FdAlg}$ of finite-dimensional 
von Neumann algebras and unital $*$-homomorphisms is one such linear category: any quantum circuit is a sequence of unitary operators followed by a measurement, so the category $\mathbf{FdHilb}$ of finite dimensional Hilbert spaces is an obvious example of the linear category $\mathbf{M}$ which forms the basis for the model of quantum programming languages like PQM in~\cite{pqm-small}. 
But there is no adjunction between $\mathbf{FdHilb}$ and $\mathbf{Set}$, so one has to embed it in an appropriate linear category to form an LNL model for PQM.  In fact, $\mathbf{FdHilb}$ embeds contravariantly in $\qCPO$, which also is a $\CPO$ enriched category; this opens the door to using $\qCPO$ as a concrete model for quantum programming languages like PQM that also supports recursion.  

The categories $\qCPO$ and $\qCPO_{\perp!}$ (the subcategory of \textbf{qCPO} of pointed objects and strict maps) are completely analogous to $\CPO$ and $\CPO_{\perp !}$ in the nonlinear setting. In fact, the diagram below consists of the left adjoints of a commuting square of symmetric monoidal adjunctions:
\begin{equation}\label{diag:lift quote}
\begin{tikzcd}
\CPO_{\perp!}
\arrow{r}{`(-)}
&
\qCPO_{\perp!}
\\
\CPO
\arrow{u}{(-)_\perp}
\arrow{r}[swap]{`(-)}
&
\qCPO
\arrow{u}[swap]{(-)_\perp}
\end{tikzcd}
\end{equation}
Moreover, $\qCPO$ is symmetric monoidal closed, and the horizontal left adjoints in the above diagram are both inclusion functors. What's more, 
the `lifting' functor $(-)_\perp:\qCPO\to\qCPO_{\perp!}$ in the diagram restricts to the ordinary lifting functor $(-)_\perp:\CPO\to\CPO_{\perp!}$. Combining this lift with the inclusion $\CPO\to\qCPO$ gives us a linear/non-linear model: 
\begin{equation}\label{qCPO model}
\begin{tikzcd}
\phantom{\mathbf{q}}\CPO\ar[r,bend left,"F",""{name=A, below}] & \qCPO_{\perp!}\ar[l,bend left,"G",""{name=B,above}] \ar[from=A, to=B, symbol=\dashv]
\end{tikzcd}
\end{equation}

We now can state our main results, Theorems \ref{thm:algebraically compact}, \ref{thm:LNL-FPC} and \ref{thm:PQM} below. Following~\cite{fiore-thesis}, we view a \emph{type constructor} as a bifunctor on the category of types, and then Theorem~\ref{thm:algebraically compact} states that all type constructors on $\qCPO_{\perp!}$ are parametrically algebraically compact. It follows that $\qCPO_{\perp,!}$ supports recursive types defined by its type constructors. Theorem~\ref{thm:LNL-FPC} then states that (\ref{qCPO model}) is a sound model for LNL-FPC (and hence supports recursive types) that is computationally adequate at nonlinear types, while Theorem~\ref{thm:PQM} says (\ref{qCPO model}) also is a sound model for Proto-Quipper-M that contains $\mathbf{FdAlg}^{\mathrm{op}}$ as a monoidal subcategory of the linear category $\qCPO_{\perp!}$. Importantly, to our knowledge, (\ref{qCPO model}) is the only known LNL model to satisfy all these properties.

\subsection{Related work} 
Typical quantum programming languages such as Proto-Quipper-M and its relatives \cite{pqm-small}, the quantum lambda calculus \cite{selingervaliron:quantumlambda}, and QWire \cite{qwire}, all have models consisting of a linear/non-linear adjunction, i.e., a monoidal adjunction as given in (\ref{diag:LNL}), where the symmetric monoidal category $\mathbf L$ contains a suitable monoidal subcategory representing quantum circuits, such as the category $\mathbf{FdAlg}$ of finite-dimensional von Neumann algebras. Models for PQM were given first in \cite{pqm-small} and then in \cite{LMZ18}. A model of the quantum lambda calculus in terms of Lafont categories (which also are linear/non-linear models \cite{mellies:categoricalmodelsofLL}) was given in  \cite{quant-semantics}. Another model was given in \cite{ChoWesterbaan16}. Models for QWire are given in \cite{ewire-lmcs}, including descriptions of their relations to linear/non-linear models. 

We discuss some of these models in more detail. The category $\CPO$ of $\omega$-complete partial orders (cpos) and Scott continuous maps, and its subcategory $\CPO_{\perp!}$ of pointed cpos and strict maps play a fundamental role in the semantics of programming languages supporting recursion~\cite{fiore-plotkin}. Using this approach, PQM was extended with term recursion in~\cite{LMZ18}, and soundness was established for the linear/nonlinear pre\-sheaf model
\begin{equation}\label{presheaf model}
    \CPO\rightleftarrows[\mathbf{FdAlg},\CPO_{\perp!}].
\end{equation}

The article \cite{LMZ20} focused on extending PQM with recursive types, but considered only the circuit-free fragment of the language, dubbed LNL-FPC, because it extends FPC with linear types. That paper also includes a computational adequacy result at non-linear types for abstract models of LNL-FPC. 

Because the tensor product of $[\mathbf{FdAlg},\CPO_{\perp!}]$ fails to satisfy a crucial hypothesis, the computational adequacy proof in~\cite{LMZ20} does not apply to this presheaf model.
In fact, the only concrete model that satisfies all the required conditions for adequacy in~\cite{LMZ20} is
\begin{equation}\label{CPO model}
\CPO\rightleftarrows\CPO_{\perp!},
\end{equation}
but this model does not support quantum circuits since $\mathbf{FdAlg}$ does not embed in $\CPO_{\perp!}$.

The model of the quantum lambda calculus given in  \cite{quant-semantics} is sound, computationally adequate at the unit type, and  supports term recursion. However, its construction is intricate, so an alternative model was proposed in \cite{ChoWesterbaan16}:
\begin{equation}\label{WStar model}
\Set\rightleftarrows\mathbf{WStar}^{\mathrm{op}}_{\mathrm{NMIU}}\rightleftarrows\mathbf{WStar}^{\mathrm{op}}_{\mathrm{NCPSU}},
\end{equation}
where $\mathbf{WStar}_{\mathrm{NMIU}}$ denotes the category of von Neumann algebras and normal $*$-homomorphisms, and $\mathbf{WStar}_{\mathrm{NCPSU}}$ denotes the category of von Neumann algebras and normal completely positive subunital maps. 
This model (\ref{WStar model}) is natural because it utilizes von Neumann algebras, which are commonly used in physics to model quantum systems. 
However, this model is enriched over $\Set$, not $\CPO$, and it can be shown that the comonad induced by the adjunction between $\Set$ and $\mathbf{WStar}_{\mathrm{NMIU}}^{\mathrm{op}}$
cannot be algebraically compact, which implies the standard proof relying on algebraically compact type constructors cannot be used to prove this model supports recursive types.

The only proof we know that an LNL model supports recursion requires  
the Cartesian closed category of the model to be $\CPO$, and  the linear category to be $\CPO$-enriched. 
To apply this approach to model (\ref{WStar model}) requires replacing $\Set$ with $\CPO$, but the functors in the model
don't admit obvious extensions. Instead, our solution is to restrict the class of von Neumann algebras in the model, and to equip them with a generalized cpo structure, which we call a quantum cpo.

\subsection{Overview of the Rest of the Paper}
The remainder of the paper lays out the theory of quantum cpos, their relation to classical cpos, and the basic categorical results that are needed to use quantum cpos as semantic models.
We begin with quantum sets \cite{Kornell18}, which describe combined classical/quantum systems that are discrete in the sense that every complete Boolean algebra of propositions is isomorphic to a power set. The results in~\cite{Kornell18} draw heavily on the work in~\cite{kuperbergweaver:quantummetrics} and \cite{Weaver10}.

We next consider quantum posets and the morphisms between them. Quantum sets can be endowed with the quantum analog of a \emph{partial order} -- a reflexive, antisymmetric transitive relation.
This notion is derived from the notion of a binary relation between quantum sets, essentially the quantum relations of Weaver \cite{Weaver10}.
By formulating an appropriate quantum generalization of families of monotonically increasing sequences, we then are able to define quantum cpos as a subclass of quantum posets, and Scott continuous functions as a subclass of monotone functions between them.

The presentation then describes the monoidal closed categories $\qCPO$ and $\qCPO_{\perp!}$. We state all the properties that make these categories useful as models for Proto-Quipper-M, and likely  for the semantics of arbitrary quantum programming languages as well. Finally, we conclude with a short discussion of intended future work. 

Because of space limitations, most proofs are omitted.

\section{Quantum Sets and Relations}\label{section qsr}
We briefly review the essential notions of \cite{Kornell18}. A quantum set is a collection of finite-dimensional Hilbert spaces, e.g., $\X = \{\CC^2, \CC^5\}$. We call these Hilbert spaces the atoms of $\X$. \emph{Intuitively}, a quantum set $\X$ is not an ordinary set, so we prefer to write $X \in \At(\X)$, rather than $X \in \X $, when $X$ is an atom of $\X$. Thus, $\X$ is a quantum set, whereas $\At(\X)$ is an ordinary set, that is, a set in the ordinary sense, although $\X$ and $\At(\X)$ are formally equal. This choice of notation affects the meaning of our terminology and expressions; for example, we will see that $\ell^\infty(\X)$ means something other than $\ell^\infty(\At(\X))$. This convention is convenient for navigating the dictionary of quantum generalizations.

Two canonical quantum sets are $\mathbf 0 = \emptyset$, and $\mathbf 1 = \{\CC\}$. The Cartesian product of two quantum sets $\X$ and $\Y$ consists of tensor product Hilbert spaces: $\X \times \Y = \{X \tensor Y \suchthat X \in \At(\X), \, Y \in \At(\Y)\}$. Each quantum set $\X$ has a dual $\X^* = \{X^* \suchthat X \in \At(\X)\}$, which consists of dual Hilbert spaces.

A quantum set with a single atom is called \emph{atomic}. A quantum binary relation from an quantum set $\{H\}$ to an atomic quantum set $\{K\}$ is just a space of operators from $H$ to $K$, i.e., a subspace of $L(H,K)$. More generally, a quantum binary relation $R$ from a quantum set $\X$ to a quantum set $\Y$ is a matrix of such subspaces; formally, $R$ is a choice of subspaces $R(X,Y) \subsetof L(X, Y)$, for each $X \in \At(\X)$, and each $Y \in \At(\Y)$.

Quantum sets and binary relations form a dagger-compact category $\qRel$. The composition of binary relations then amounts to matrix multiplication, and as a matrix, the identity binary relation $I$ has subspaces $\CC \cdot 1$ down the diagonal, and zero subspaces off the diagonal. The adjoint binary relation $R^\dagger$ from $\Y$ to $\X$ is defined by $R^\dagger(Y,X) = \{ F^\dagger\mid F\in R(X,Y)\}$, for $X\in \At(\X)$ and $Y\in \At(\Y)$. 

The operator spaces $L(X,Y)$ are Hilbert spaces, so the binary relations from $\X$ to $\Y$ form an orthomodular lattice \cite{Sasaki54}. The bottom element of this orthomodular lattice is $\bot_{\X, \Y}$, defined by $\bot_{\X,\Y}(X, Y) = 0$. The top element of this orthomodular lattice is $\top_{\X,\Y}$, defined by $\top_{\X,\Y}(X, Y) = L(X,Y)$.
Two binary relations $R$ and $S$ are said to be \emph{orthogonal}, synonymously \emph{disjoint}, just in case they are orthogonal in each entry, i.e., just in case $\mathrm{Tr}_X (S(X, Y)^\dagger \cdot R(X,Y)) = 0$ for all $X \in \At(\X)$ and $Y \in \At(\Y)$.

Each ordinary set $S$ is identified with a quantum set $`S$, obtained by replacing each element of $S$ with a one-dimensional Hilbert space. The space of linear operators from one one-dimensional Hilbert space to another is itself one-dimensional, so a quantum binary relation from $`S$ to $`T$ is corresponds to an ordinary binary relation from $S$ to $T$. Thus, we have an inclusion functor from the category of ordinary sets and ordinary binary relations, to the category of quantum sets and quantum binary relations. It preserves the obvious dagger-compact structure on the former category, and it is essentially surjective onto those quantum sets whose atoms are all one-dimensional, which we term \emph{classical}.

A \emph{function} $F\colon \X\to \Y$ is simply a binary relation $F$ from $\X$ to $\Y$ satisfying $F \circ F^\dagger \leq I$ and $F^\dagger \circ F \geq I$; these are the same conditions that characterize which ordinary relations $f$ are functions, if one interprets $F^\dagger$ as the converse $f^{-1}$. Together with the Cartesian product of quantum sets defined above, the category $\qSet$ of quantum sets and functions is monoidal closed. Then the inclusion functor from the category $\Set$ of ordinary sets and ordinary functions into $\qSet$ is strong monoidal, and thus $\Set$ and $\qSet$ form a linear/non-linear model~\cite{benton-big,benton-small} with $\qSet$ as the linear category.

The category of quantum sets and functions is also dual to a category $\Alg$ of operator algebras. Up to isomorphism, the objects of $\Alg$ are von Neumann algebras of the form $\bigoplus_{\alpha \in I} M_{n_\alpha}(\CC)$, and the morphisms of $\Alg$ are unital normal $*$-homomorphisms. The contravariant equivalence takes each quantum set $\X$ to the von Neumann algebra $\ell^\infty(\X) := \bigoplus_{X \in \At(\X)} L(X)$.
Each function $F\colon \X\to\Y$ corresponds to a unital normal $*$-homomorphism $F^\star\colon \ell^\infty(\Y)\to \ell^\infty(\X)$.

Normal states on $\ell^\infty(\X)$ are intuitively probability distributions on $\X$. A normal state $\mu$ on $\ell^\infty(\X)$ can be pushed forward along a function $F$ from $\X$ to $\Y$, using the contravariant equivalence of the previous paragraph. Explicitly, the pushforward of $\mu$ is the normal state $\mu \circ F^\star$ on $\ell^\infty(\Y)$. If $\X = `S$ and $\Y = `T$, for sets $S$ and $T$, then this corresponds to the usual pushforward of a probability distribution.

A normal state $\mu$ on $\ell^\infty(\X)$ can also be described by an element $m \in \ell^1(\X) := \bigoplus_{X \in \At(\X)} L(X)$, using the expected formula $\mu(a) = \Tr(a\cdot m)$, for $a \in \ell^\infty(\X)$. This trace is defined by  $\Tr(a) = \sum_{X \in \At(\X)} \Tr(a(x))$, for $a \in \ell^\infty(\X)$. The operator $m$ is a choice of density matrices, but normalized so that their traces sum to $1$. The pushforward of $\mu$ can be calculated from $m$ using Kraus operators. The contribution of the density matrix of $\mu$ at the atom $X \in \At(\X)$, to the density matrix of $\mu \circ F^\star$ at the atom $Y \in \At(\Y)$ is given by the expression $\dim(Y) \cdot \sum_{v \in B} v \cdot m(X) \cdot v^\dagger$, where $B$ is any basis for $F(X,Y)$. Indeed, $F(X,Y)$ is canonically a Hilbert space for the inner product $\langle f_1 | f_2 \rangle = \Tr(f_1 \cdot f_2^\dagger)$.

\section{Modeling physical systems}
Finitary physical types are naturally modelled by quantum sets. For example, the qubit is modelled by the quantum set $\H_2$, whose only atom is a two-dimensional Hilbert space. By contrast, the classical bit is modelled by the quantum set $`\{0,1\}$, which has two one-dimensional atoms. The `memory' of an idealized quantum computer might consist of finitely many qubits, and finitely many bits, so we can model such a quantum computer as a composite physical system.

The Cartesian product of ordinary sets generalizes to a symmetric monoidal structure $\times$ on the categories $\mathbf{qRel}$ and $\mathbf{qSet}$. We recall that the monoidal product of two quantum sets is defined by forming binary tensor products of their atoms. Composite systems consisting of two fully quantum systems, that is, of quantum systems each modelled by a single Hilbert space, are obtained by forming the tensor product of the two Hilbert spaces. Similarly, composite systems consisting of two fully classical systems, that is, of quantum systems modelled by ordinary sets, are obtained by forming the ordinary Cartesian product of the two sets. Thus, our generalized product models composite systems consisting entirely of fully quantum systems, or entirely of fully classical systems. In fact, it is appropriate for modelling mixed quantum/classical systems as well. In particular, an idealized quantum computer possessing $n$ qubits and $m$ bits can be modelled by the quantum set $ \underbrace{\H_2 \times \cdots \times \H_2}_{n} \times \underbrace{`\{0,1\} \times \cdots \times `\{0,1\}}_{m}.$
\vspace{-2.5ex}
\noindent This quantum set has $2^m$ atoms, each of dimension $2^n$.

Other natural type constructors are also easily modelled in quantum sets. Sum types are modelled by disjoint unions of quantum sets, portraying a kind of classical disjunction of potentially quantum systems. The resulting physical system may be in a pure state of the first type or in a pure state of the second type, but no superpositions may occur. For example, an idealized quantum computer possessing a single qubit and a single bit is modelled by a quantum set with two atoms, because the computer can be in a configuration where the bit has value $0$, and it can be in a configuration where the bit has value $1$, but it cannot be in a superposition between two such states. 

Higher types are modelled by quantum function sets \cite[Definition 9.2]{Kornell18}, the inner hom objects of the closed monoidal category $\qSet$. These quantum sets are more challenging to describe, mainly due to the large automorphism group of the physical qubit. The one-dimensional atoms of the quantum function set $\H_2^{\H_2}$ are in canonical bijective correspondence with the automorphisms of $\H_2$ in $\qSet$, i.e., with the $*$-automorphisms of $M_2(\CC)$. The quantum function set $\H_2^{\H_2}$ has higher-dimensional atoms as well. In general, the quantum function set from a quantum set $\X$ to a quantum set $\Y$ has an atom of dimension $d$ for each unital normal $*$-homomorphism $\rho\: \ell^\infty(\Y) \To \ell^\infty(\X) \overline \otimes L(H_d)$\footnote{Here $\overline \otimes$ denotes the spatial tensor product of von Neumann algebras \cite[Definition IV.1.3]{takesaki:oa1}.} that is irreducible in the sense that $0$ and $1$ are the only projections of the form $1 \tensor p$ that commute with the image of $\rho$.

The bang operator $!$ applied to a quantum set extracts the quantum subset consisting of its one-dimensional atoms; this is essentially an ordinary set. This also is the largest subset of the original quantum set that admits a duplication map. Indeed, classical sets all admit duplication maps, but the no-cloning theorem forbids the duplication of higher dimensional atoms. Returning to the example of the previous paragraph, every atom of $!(\H_2^{\H_2})$ is one-dimensional, and furthermore $!(\H_2^{\H_2}) \iso `\mathbf{qSet}(\H_2, \H_2)$.

Each quantum gate or measurement may be modelled by a function between quantum sets. Such a function is analogous to a function between the configuration spaces of two classical systems. A function between configuration spaces induces an map from the states on the domain system to the states on the codomain system. The function pushes probability measures on the domain system forward onto the codomain system. The standard formalizations of quantum gates and measurements as maps on state spaces are analogous to these pushforward maps. Examples \ref{Hadamard} and \ref{measurement} below exhibit the formalizations of the Hadamard gate and of qubit measurement respectively by functions between quantum sets. Each function induces a map on states, and in Example \ref{channel} below, we examine this map in the case of qubit measurement, recovering the expected probability distributions on experimental outcomes.

\begin{example}\label{Hadamard}
Quantum gates are automorphisms of fully quantum systems, which consist of qubits. Such an automorphism is typically formalized by a unitary operator. In \textbf{qSet}, such an automorphism is formalized by a closely-related function. Indeed, each function $F_1\: \H_d \To \H_d$, for any positive integer $d$, is defined by $F_1(H_d, H_d) = \CC\cdot u$, for some unitary operator $u$. For example, the Hadamard gate is formalized by the function $F_1\: \H_2 \To \H_2$ defined by $$F_1(H_2, H_2) = \CC \cdot \left[ \begin{matrix} 1 & 1 \\ 1 & -1 \end{matrix} \right].$$
This function $F_1$ is an automorphism of $\H_2$ in $\qSet$. 
The unital normal $*$-homomorphism $M_2(\CC) \iso \ell^\infty(\H_2) \To \ell^\infty(\H_2) \iso M_2(\CC)$ that corresponds to this function, in the sense of the duality between quantum sets and hereditarily atomic von Neumann algebras \cite{Kornell18}, is simply conjugation by the Hadamard matrix above, appropriately normalized.

\end{example}

\begin{example}\label{measurement}
Measurement is a channel from a fully quantum system to a fully classical system; it may be formalized by a function from an atomic quantum set to a classical quantum set. The effect of this function on states yields probability distributions on experimental outcomes. In this example we exhibit the function, and in the next example, we examine its effect on states.

The definition of a function between quantum sets implies functions from an atomic quantum set $\H_d$ to a classical quantum set $`S$ are in bijective correspondence with projection-valued measurements on $\H_d$. Explicitly, $F_2(H_d, \CC_s) = L(H_d, \CC_s) \cdot  p_s$ for each $s \in S$, where $(p_s\suchthat s \in S)$ is a projective POVM.\footnote{POVM stands for \emph{Positive Operator-Valued Measurement}.} In particular, the standard measurement of a qubit is formalized by a function $F_2\colon \H_2\to `\{1, -1\}$ defined by
$$F_2(H_2, \CC_s) = \begin{cases} \CC \cdot [\begin{matrix} 1 & 0 \end{matrix} ] & {s = 1}; \\ 
\CC \cdot [\begin{matrix} 0 & 1 \end{matrix}]  & s = {-1}.    \end{cases}$$
This function $F_2$ is an epimorphism in $\qSet$. Equivalently, it is surjective in the sense that $F_2 \circ F_2^\dagger \geq I$.
The unital normal $*$-homomorphism $\CC^2 \iso \ell^\infty(\{1, -1\}) \To \ell^\infty(\H_2) \iso M_2(\CC)$ corresponding to this function, in the sense of the duality between quantum sets and hereditarily atomic von Neumann algebras \cite{Kornell18}, includes $\CC^2$ into $M_2(\CC)$ diagonally. \end{example}

\begin{example}\label{channel}
A function between quantum sets is also a quantum channel in the familiar mathematical sense: normal states are pushed forward to normal states, as described in section \ref{section qsr}, and the function is completely determined by this mapping. In the case of the measurement channel,  $\{[\begin{matrix} 1 & 0 \end{matrix} ]\}$ is an orthonormal basis for $F_2(H_2, \CC_1)$, and  $\{[\begin{matrix} 0 & 1 \end{matrix} ]\}$ is an orthonormal basis for $F_2(H_2, \CC_{-1})$. A state on $\ell^\infty(\H_2) = M_2(\CC)$ is a density matrix $m = \left[\begin{matrix} m_{11} & m_{12} \\ m_{21} & m_{22}\end{matrix}\right]$, while a state on $\ell^\infty(`\{1,-1\}) \iso \CC \oplus \CC$ is a pair of positive numbers $p_1$ and $p_{-1}$ that sum to $1$. Pushing forward the density matrix $m$ in the prescribed way, we find that 
$$ 
p_1 = \dim(\CC_1) \cdot [\begin{matrix} 1 & 0 \end{matrix} ] \cdot\left[\begin{matrix} m_{11} & m_{12} \\ m_{21} & m_{22}\end{matrix}\right]\cdot \left[\begin{matrix} 1 \\ 0 \end{matrix} \right] = m_{11}, ~~ \text{and} ~~
p_2 = \dim(\CC_{-1}) \cdot [\begin{matrix} 0 & 1 \end{matrix} ] \cdot\left[\begin{matrix} m_{11} & m_{12} \\ m_{21} & m_{22}\end{matrix}\right]\cdot \left[\begin{matrix} 0 \\ 1 \end{matrix} \right] = m_{22}.
$$
In this way we recover the probabilities of our two experimental outcomes from $F_2$.

\end{example}

\section{Quantum posets}

\begin{definition}

A \emph{quantum poset} is a pair $(\X,R)$ consisting of a quantum set $\X$ and a relation $R\in\qRel(\X,\X)$ satisfying: (1) $I_\X\leq R$ (reflexivity); (2) $R\circ R\leq R$ (transitivity); and (3)  $R\wedge R^\dag \leq I_\X$ (antisymmetry).
(This definition of a quantum poset is essentially that of Weaver \cite[Definition 2.6]{Weaver10}.)

A \emph{monotone} map $F:(\X,R)\to(\Y,S)$ is simply a function $F:\X\to\Y$ satisfying $F\circ R\leq S\circ F$.\footnote{This is analogous to the condition $(f\times f)\circ (\leq_P)\ \subseteq\ (\leq_Q)\circ (f\times f)$ which characterizes when a map $f\colon (P,(\leq_P))\to (Q,(\leq_Q))$ between ordinary posets is monotone.} We denote the category of quantum posets with monotone maps by $\qPOS$.
\end{definition}

\begin{example}
Let $\X$ be a quantum set. Then $I_\X$ is a quantum order on $\X$, which we call the \emph{trivial} order.
\end{example}

If $(S,\sqsubseteq)$ is an ordinary poset, then $(`S,`\!\!\sqsubseteq)$ is a quantum poset, and vice versa; for example, the trivial order $\sqsubseteq$ on $S$ corresponds to $`(\sqsubseteq)\,=\, I_{`S}$, the trivial order on $`S$. And, a monotone map $f$ between ordinary posets gives rise to a monotone function  $`f$ between the associated quantum posets, and vice versa. It follows that $`(-)$ extends to a fully faithful functor $\POS\to\qPOS$.

\begin{example}\label{ex:non-classical quantum poset}
A  `non-classical' quantum poset is given by the relation $R$ on $\H_2$ specified by \[
R(H_2,H_2)=\CC\begin{pmatrix}
1 & 0 \\
0 & 1
\end{pmatrix} +\CC \begin{pmatrix}
0 & 1 \\
0 & 0
\end{pmatrix}.\]
Since $\H_2$ has only one atom $H_2$, $R$ is determined by $R(H_2,H_2)$. Then $(\H_2,R)$ is a quantum poset. 
\end{example}

The partial orders that appear in recursion theory are often viewed as \emph{information orderings}. For example, a partial function on the natural numbers can be viewed as one step in the construction of a total function, i.e., as a partially specified total function. In this sense, an extension of a given partial function carries more information about the total function that it purports to describe.

A preorder structure on the phase space of a classical system can be viewed as formalizing some entropic process undergone by the system. More precisely, whereas a stochastic matrix records the transition probabilities of the process, the preorder structure simply records which transitions are possible. A lower configuration in the order is the result of more transitions, and thus carries less information about the initial configuration of the system. 

The same intuition is available in the quantum setting. A preorder structure on a quantum set $\X$ may be understood as an information order on that quantum set not only by analogy with the classical case, but also directly. A stochastic channel on $\H_d$, formally a completely positive map of the appropriate kind, determines a preorder on $\H_d$, essentially the unital subalgebra of $L(H_d)$ generated by the adjoints of its Kraus operators (c.f.~\cite{Weaver19}). We work with the algebra generated by the adjoints because the information order is opposite to the transition order; later states carry less information.

Any partial order on a finite set can be encoded as the possible transitions of some stochastic matrix, but the same is not true in the quantum case. For instance, it is easy to show that the partial order on $\H_2$ defined in Example \ref{ex:non-classical quantum poset} does not arise in this way. This phenomenon is intuitively related to the fact that some partial functions between quantum sets cannot be extended to total functions. If we allow partial channels, which are permitted to consume their input without producing an output, then the partial order in \ref{ex:non-classical quantum poset} arises from the partial channel $\Phi$:
$$\Phi(\rho) = \frac 1 2 \rho + \frac 1 2 v \rho v^* \qquad \text{where} \qquad v =
\left(
\begin{matrix}
0 & 0 \\
1 & 0
\end{matrix}
\right)
$$

In both the classical and quantum cases, it is appropriate to identify configurations that are equivalent in the information they carry when modelling recursion. In the classical case, we turn the preordered set into a partially ordered set by taking a quotient. This quotient construction has a quantum analog, which is most easily viewed using Weaver's characterization of binary relations between quantum sets \cite{Weaver10}.

Thus, quantum posets, and therefore also quantum cpos, can be regarded as quantum systems equip\-ped with a stochastic channel that models a kind of decay, and therefore a kind of information order. 

More elaborate examples of quantum posets can be described using the coproduct and the monoidal product on $\qPOS$.
The monoidal product $(\X,R)\times(\Y,S)$ of $(\X,R)$ and $(\Y,S)$ in $\qPOS$ is the quantum poset $(\X\times\Y,R\times S)$, where $\X\times\Y$ is the monoidal product of objects in $\qRel$, and $R \times S$ is the monoidal product of morphisms in $\qRel$. The embedding $`(-):\POS\to\qPOS$ is strong monoidal, and has a right adjoint.

Similarly, the coproduct $(\X,R)\uplus(\Y,S)$ of $(\X,R)$ and $(\Y,S)$ is the quantum poset $(\X\uplus\Y,R\uplus S)$, where $\X\uplus \Y$ is the coproduct of objects in $\qRel$, as well as in $\qSet$, and $R\uplus S$ is the coproduct of morphisms in $\qRel$. The generalization to coproducts of arbitrary families is immediate. 

The limit of a diagram $D:A\to\qPOS$ can be formed as follows: For each $\alpha\in A$, let $D(\alpha)=(\X_\alpha,R_\alpha)$. Let $\X$ be the limit of the $\X_\alpha$ in $\qSet$, and let $J_\alpha:\X\to\X_\alpha$ be the limiting maps. Then $R=\bigwedge_{\alpha\in A}J_\alpha^\dag\circ R_\alpha\circ J_\alpha$ is a partial order on $\X$, and $(\X,R)$ is the limit of $D$ in $\qPOS$.

A partial order $R$ on a quantum set $\X$ also imposes a partial order on the set of functions into $\X$ from any fixed quantum set $\W$: if $F$ and $G$ are functions $\W \To \X$, we define $F \sqsubseteq G$ iff $G \leq R \circ F$; this is completely analogous to the classical setting where $g\subseteq \leq \circ f$ expresses $g\leq f$. 

\begin{lemma}\label{lem:order enrichment}
Let $(\X,R)$ be a quantum poset, and let $\W$ be any quantum set. Then set $\qSet(\W, \X)$ of functions from $\W$ to $\X$ is partially ordered by $F\sqsubseteq G$ iff $G \leq R \circ F$.
\end{lemma}

It follows that $\qPOS(\X,\Y)$ inherits the order from its superset $\qSet(\X,\Y)$ for any quantum posets $\X$ and $\Y$. This order does not depend on the order on $\X$, which is also what happens classically. 
\begin{lemma}
$\qPOS$ is order enriched.
\end{lemma}

The category $\textbf{qPOS}$ has all the categorical properties of the category $\textbf{POS}$:
\begin{theorem}\label{thm:qPOS properties}
The category $\qPOS$ is complete, has all coproducts, and is monoidal closed. The functor $`(-):\POS\to\qPOS$ is strong monoidal and has a right adjoint.
\end{theorem}
Since $\qPOS$ is monoidal closed, it is enriched over itself, so $\times$ is a $\qPOS$-bifunctor. By Theorem \ref{thm:qPOS properties}, $`(-):\POS\to\qPOS$ is the left adjoint of an LNL model, and it follows that $\times$ is also a $\POS$-functor. One also can show that the monoidal product $\times$ 
reflects the order on homsets, hence:
\begin{proposition}\label{prop:tensor product reflects order}
Let $(\X_1,R_1)$, $(\X_2,R_2)$, $(\Y_1,S_1)$ and $(\Y_2,S_2)$ be quantum posets, and let $F_1,G_1:\X_1\to\Y_1$ and $F_2,G_2:\X_2\to\Y_2$ be functions. Then both $F_1\sqsubseteq G_1$ and $F_2\sqsubseteq G_2$ if and only if $F_1\times F_2\sqsubseteq G_1\times G_2$.
\end{proposition}

\section{Quantum CPOs}
An $\omega$-CPO is an ordinary poset $P$ satisfying the property that any increasing chain $x_1\leq x_2\leq\cdots \in P$ has a \emph{least upper bound} $x_\infty\in P$ satisfying $x_\infty = \sup_n x_n$. It is well known that this is equivalent to the condition that, for any ordinary set $W$ and any increasing sequence of functions $f_1\leq f_2\leq\cdots \colon W\to P$, there is a function $f_{\infty} \colon W\to P$ satisfying $\sup_n f_n(w) = f_\infty(w)$, $\forall w\in W$. This motivates the following definitions:

\begin{definition}\label{nearrow}
Let $(\X,R)$ be a quantum poset, and let $\W$ be any quantum set. Let $K_1\sqsubseteq K_2\sqsubseteq\cdots:\W\to\X$ be a monotonically increasing sequence for the order on $\qSet(\W,\X)$ defined above.
If there exists a function $K_\infty:\W\to\X$ such that
\[ R\circ K_\infty =\bigwedge_{n\in\NN}R\circ K_n,\]
then we say that $K_\infty$ is the \emph{limit} of the sequence, and we write $K_n\nearrow K_\infty$.
\end{definition}
 
Thus, it would be natural to define a quantum cpo to be a quantum poset $\X$ that has a limit $K_\infty$ for any increasing sequence $K_1 \sqsubseteq K_2 \sqsubseteq \cdots\colon \W\to \X$, from any quantum set $\W$. To tame this characterization, we recall $\H_n$ denotes the quantum set whose single atom is $H_n$, the $n$-dimensional Hilbert space. Together the quantum sets $\H_1,\H_2, \H_3, \ldots$ form a generating family \cite[Definition 4.5.1]{borceux:handbook1} for $\qSet$, leading to:

\begin{definition}\label{def:quantcpo}
\begin{enumerate}
\item[(a)] A quantum poset $(\X,R)$ is a \emph{quantum cpo} iff for each $d\in\NN$, and each increasing sequence  $K_1\sqsubseteq K_2\sqsubseteq\cdots \colon \H_d \To \X$, there is a function $K_\infty\: \H_d \To \X$ such that $K_n \nearrow K_\infty$.
\item[(b)] A function $F\: \X \To \Y$  between quantum cpos  $(\X,R)$ and $(\Y,S)$ is \emph{Scott continuous} iff for each $d \in \NN$, and all functions $K_i\colon \H_d \To \X$, with $i \in \NN \union \{\infty\}$, if $K_n \nearrow K_\infty$ then $F \circ K_n \nearrow F \circ K_\infty$.
\end{enumerate}
We write $\qCPO$ for the category of quantum cpos and Scott continuous functions.
\end{definition}

As implied by the terminology, every ordinary cpo is also a quantum cpo, and conversely.

\begin{theorem}\label{classical quantum cpos}
Let $(S,\sqsubseteq)$ be an ordinary poset. Then $(`S,`\sqsubseteq)$ is a quantum cpo if and only if $(S,\sqsubseteq)$ is a cpo. Moreover, the functor $`(-):\CPO\to\qCPO$ is fully faithful.
\end{theorem}

As in the classical case, all finite quantum posets are quantum cpos.

\begin{proposition}\label{finite quantum cpos}
Let $(\X,R)$ be a quantum poset, where $\X$ has a finite number of atoms. Then $(\X,R)$ is a quantum cpo.
\end{proposition}

As a consequence, the quantum poset in Example \ref{ex:non-classical quantum poset} is a quantum cpo.
We can obtain a more complex example of a quantum cpo as a quantum analog of the powerset under the inclusion order:

\begin{example}\label{ex:monoidal closure}
Let $n$ be a positive integer, and equip the atomic quantum set $\H_n$ with the  trivial order $I_{\H_n}$, and the ordinary set $\Omega = \{0,1\}$ the standard order $ 0 < 1$, which we denote $\sqsubseteq$. Then the quantum poset $[\H_n, `\Omega]_{\sqsubseteq}$ is a quantum cpo.
\end{example}

Just as $\textbf{POS}$ and $\textbf{qPOS}$ form an LNL model, so do $\textbf{CPO}$ and $\textbf{qCPO}$:
\begin{theorem}\label{qCPO theorem}
The category $\qCPO$ is monoidal closed, complete, has all coproducts, and is enriched over $\CPO$ as a monoidal closed category, and the functor $`(-):\CPO\to\qCPO$ is strong monoidal and has a right adjoint. 
\end{theorem}

Classical cpos and finite quantum posets are the basic examples of quantum cpos. Theorem \ref{qCPO theorem} allows us to combine these basic examples to form new quantum cpos.

\subsection{Pointed quantum cpos}
Any model of a programming language supporting recursion needs denotations for nonterminating terms. In ordinary domain theory, this is done by passing to the category $\CPO_{\perp!}$ of pointed cpos with strict Scott continuous maps. In this section we formulate the quantum generalization of this category. 

We say that a quantum cpo $(\X,R)$ is \emph{pointed} if there exists a unique one-dimensional atom $X^\bot\in \At(\X)$ such that $R(X^\bot,X)=L(X^\bot,X)$ for each $X\in \At(\X)$. It follows that any ordinary cpo $(P,\sqsubseteq)$ is pointed if and only if $(`P,`\sqsubseteq)$ is pointed. Given pointed quantum cpos
$(\X,R)$ and $(\Y,S)$, a Scott continuous map $F:\X\to\Y$ is said to be \emph{strict} iff $F(X^\bot,Y^\bot)=L(X^\bot,Y^\bot)$. 
We denote the category of pointed quantum cpos with strict Scott continuous maps by $\qCPO_{\perp!}$. 

There also is a lift functor $(-)_\perp:\qCPO\to\qCPO_{\perp!}$, generalizing the ordinary lift functor in the sense that Diagram \ref{diag:lift quote} of Section \ref{sec:intro} commutes. Just as the ordinary lift functor is left adjoint to the inclusion of $\CPO_{\perp!}$ into $\CPO$, we have
\begin{lemma}
The lift functor $(-)_{\perp}:\qCPO\to\qCPO_{\perp!}$ is left adjoint to the inclusion functor $\qCPO_{\perp!}\to\qCPO$.
\end{lemma}

There is a symmetric monoidal product $\otimes$ on $\qCPO_{\perp!}$, reminiscent of the smash product on $\CPO_{\perp!}$. The monoidal unit is obtained by lifting the monoidal unit $\mathbf 1$ of $\qCPO$.

\begin{proposition}
The category $\qCPO_{\perp!}$ is monoidal closed, and the lifting functor $(-)_\perp:\qCPO\to\qCPO_{\perp!}$ is strong monoidal. 
\end{proposition}

Proposition \ref{prop:tensor product reflects order} implies the monoidal product on $\qCPO_{\perp!}$ reflects the order on  non-zero morphisms with domain $\mathbf 1_\perp$, the monoidal unit of $\qCPO_{\perp!}$.
\begin{proposition}\label{prop:tensor product in qCPObs reflects order}
Let $(\X_1,R_1)$ and $(\X_2,R_2)$ be pointed quantum cpos, and let $F_1,G_1:\mathbf 1_\perp\to\X_1$ and $F_2,G_2:\mathbf 1_\perp\to\X_2$ be strict Scott continuous functions such that $F_i\neq 0_{\mathbf 1_1\X_i}$ for $i=1,2$. Then both $F_1\sqsubseteq G_1$ and $F_2\sqsubseteq G_2$ if and only if $F_1\otimes F_2\sqsubseteq G_1\otimes G_2$. 
\end{proposition}
Thus the monoidal product on $\qCPO_{\perp!}$ behaves similarly to the smash product on $\CPO_{\perp!}$.
Our main theorem provides the LNL model defined by $\textbf{CPO}$ and $\textbf{qCPO}_{\perp !}$:

\begin{theorem}\label{thm:main theorem}
The category $\qCPO_{\perp!}$ has all coproducts, is complete, symmetric monoidal closed, is enriched over $\CPO_{\perp!}$, and has a zero object that is e-initial. Furthermore both compositions in Diagram (\ref{diag:lift quote}) of Section~\ref{sec:intro} form the left adjoint component of the model in Diagram (\ref{qCPO model}). 
\end{theorem}

\section{Models of quantum programming languages}
Proto-Quipper-M (PQM) is a circuit description language designed to support writing quantum circuits in a high-level functional language and then displaying the evolution of the circuit as a diagram. The ``computer" for which PQM is intended is Knill's QRAM model -- a classical computer with a quantum coprocessor.  Here's an example of the sort of circuit PQM is meant to define. We present the circuit in $\textbf{qSet}$ for simplicity. It can then be imported intact into $\textbf{qCPO}$ by the functor $\X \mapsto (\X,I_\X)$, where $I_\X$ is the trivial order, and then into $\textbf{qCPO}_{\perp !}$ using the lift functor. 

\begin{example}
The quantum Fourier transform is a channel from an $n$-qubit system to itself, and it is given by a sequence of basic quantum gates: the Hadamard gate and rotation gates \cite[section 5.1]{NielsenChuang}. For fixed $n$, the quantum Fourier transform is thus modelled by a function on the quantum set $\H_2^n = \H_2 \times \cdots \times \H_2$ that can be expressed as a composition of these basic gates.

\begin{figure}[t!]
\begin{center}
\includegraphics[width=.20\textwidth,angle=-90]{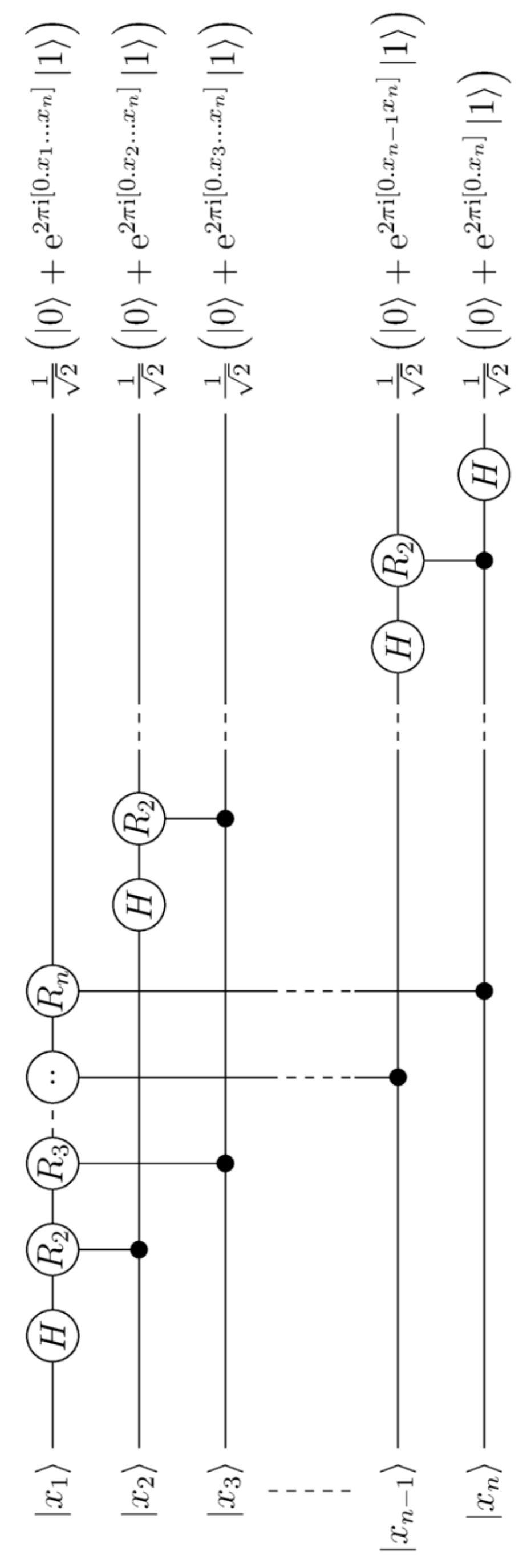}
\caption{The Quantum Fourier Transform~ \cite{wikipedia}}
\end{center}
\end{figure}

We model a quantum system consisting of an unknown number of qubits by the quantum set $\X = \mathbf 1 + \H_2 + (\H_2 \times \H_2) + \cdots$. Indeed, a pure  state on $\ell^\infty(\X)$ is just a pure state on $\ell^\infty(\H_2^n)$ for some natural number $n$. A mixed state may be supported on more than one summand, expressing uncertainty about just how many qubits there are! Note that $\X$ is a recursive type over $\H_2$; it is the list type. We have a function $N\: \X \To `\NN$ that maps each atom to its dimension; in other words, $N$ is defined by $N(H_2^{\tensor n}, \CC_n) = L(H_2^{\tensor n}, \CC_n)$, for all $n\in \NN$, with the other components vanishing. It does not play a direct role in the following description of the quantum Fourier transform.

Using the standard characterization of $\X$ as a fixed point, i.e., $\X \iso \mathbf 1 + \X \times \H_2$, we may define the term $x: \X \vdash \mathrm{Swap}(x):\X$, which reverses the order of the list, in the usual way. Suppressing this isomorphism, we ``unfold'' $x: X$, writing it as either $x = \ast: \mathbf 1$ or $x = (x_1, y): \X \times \H_2$, and then we define $\mathrm{Swap}(\ast) = \ast$ and $\mathrm{Swap}(x_1, y_1) = (y_1, \mathrm{Swap}(x_1))$. We curry to obtain a function $\mathrm{Swap}: \X \multimap \X$.

Using the same approach, we can define a function $\mathrm{Step}: \X \multimap \X$ modelling each computational step of the algorithm. To formalize this definition, we assume a controlled phase gate $\mathrm{R}: \H_2 \times \H_2 \times `\NN \multimap \H_2 \times \H_2 $, which rotates the phase of the up state of the second qubit by $2\pi/2^n$, whenever the first qubit is in the up state, or equivalently, vice versa. As before, we unfold $x : \X$. If $x = \ast$, we define $\mathrm{Step}(x) = \ast$. If $x = (x_1, y)$, then we unfold $x_1$. If $x_1 = \ast$, then we define $\mathrm{Step}(x) = Hy$, where $H\: \H_2 \multimap \H_2$ is the Hadamard gate, and if $x_1 = (x_2, y_2)$, then we define $\mathrm{Step}(x)$ to be the result of applying a controlled phase gate rotating by $2\pi/2^n$ to $(\mathrm{Step}(x_2),y_2)$. The dimension $n$ is classical data, so it can be duplicated and used any number of times at each recursion step, but $\X$ is a quantum type, so \emph{a priori}, we cannot evaluate the function $N$ without consuming the quantum data in any given step in the recursion. For this reason, this recursion should occur over $\X \times `\NN$, so we can increment the dimension separately.

The swapped quantum Fourier transform $\mathrm{QFT_0}\: \X \multimap \X$, depicted in the diagram, can now be similarly defined for $x : \X$ distinct from $\ast$, by \[\mathrm{QFT_0}(x) = (\tilde y_1, \mathrm{QFT_0}(\mathrm{Swap}(\tilde x_1)),\] where $(\tilde x_1, \tilde y_1) = \mathrm{Swap}(\mathrm{Step}(x))$. Of course, the quantum Fourier transform itself is just $\mathrm{QFT} = \mathrm{Swap} \circ \mathrm{QFT_0}$.

\end{example}

There are some important features that this example illustrates. First, lists of qubits are needed to model the memory of the quantum device, so any model for this circuit must have support for (at least) inductive types. Second, we don't know the number of qubits in the state on which the program will run when the circuit is written, so it must be able to respond to any input length. Since the QFT itself is inherently (primitive) recursive, the model must support primitive recursion. While these are the simplest forms of recursion, they make clear the need for these features. Finally, our model also includes $\textbf{FdHilb}$, which allows us to reason concretely about the effect of QFT on a given list of qubits. 

Returning to our results, we first consider type recursion.  Models built over $\mathbf{CPO}$ support type recursion if all the type constructors are \emph{algebraically compact}~\cite{fiore-thesis}. 
A direct consequence of Theorem \ref{thm:main theorem} is that the $\CPO$/$\qCPO$ model (\ref{qCPO model}) in Section \ref{sec:intro} satisfies all conditions for a $\CPO$-LNL model (cf.~\cite[Definition 5.3.1]{LMZ20}), hence we can apply \cite[Theorem 3.2]{LMZ18} and \cite[Theorem 5.3.3]{LMZ20} to conclude:
\begin{theorem}\label{thm:algebraically compact}
All $\CPO$-endofunctors on $\qCPO_{\perp!}$ are algebraically compact. In particular, $\otimes$, $+$, $\multimap$ and $!$ are $\CPO$-functors, hence endofunctors on $\qCPO_{\perp!}$ that are compositions of these four functors are algebraically compact.
\end{theorem}
Moreover, by \cite[Theorem 6.3.9]{LMZ20} any $\CPO$-LNL model is a sound model for LNL-FPC, i.e., the language introduced in \cite{LMZ20} that can both be seen as an FPC with linear types and as an extension of the circuit-free fragment of PQM with recursive types. By \cite[Theorem 7.0.10]{LMZ20} any $\CPO$-LNL model is computationally adequate at non-linear types if $I\ncong 0$ (which is the case for $\qCPO_{\perp!}$ the zero object is $\mathbf 1$, whereas the monoidal unit $I$ is $\mathbf 1_\perp$) and if the monoidal product reflects the order as in Proposition \ref{prop:tensor product in qCPObs reflects order}. So we conclude:
\begin{theorem}\label{thm:LNL-FPC}
The linear/non-linear adjunction (\ref{qCPO model}) is a sound model for LNL-FPC that is computationally adequate at non-linear types.
\end{theorem}

We now consider term recursion. PQM (including circuits) is extended with term recursion in \cite{LMZ18}. We can add the following result: 

\begin{theorem}\label{thm:PQM}
The linear/non-linear adjunction (\ref{qCPO model}) is a sound model for PQM with term recursion.
\end{theorem}
\begin{proof}
By Theorem 3.2 and Definition 3.4 of \cite{LMZ18} a model for PQM consists of a symmetric monoidal category $\mathbf M$ of circuits (for quantum computing one typically chooses this category to be  $\mathbf{FdAlg}^{\mathrm{op}}$), and an LNL-adjunction of the form (\ref{diag:LNL}) such that $\mathbf C$ and $\mathbf L$ have finite coproducts, and such that there exists a strong monoidal embedding $E:\mathbf M\to\mathbf L$. By \cite[Theorem 4.3]{LMZ18} such a model is a sound model for PQM with term recursion if the functor $!(-)$ is parametrically algebraically compact, i.e., for fixed $A\in\mathbf L$, the endofunctor $A\otimes !(-):\mathbf L\to\mathbf L$ is algebraically compact. For $\qCPO_{\perp!}$, the latter condition follows from Theorem \ref{thm:algebraically compact}. Finally, we also know $\mathbf{FdAlg}^{\mathrm{op}}$ embeds into $\qCPO_{\perp!}$.
\end{proof}

To our knowledge, the only other model for which the conclusions of Theorem \ref{thm:PQM} hold is (\ref{presheaf model}), whereas the only other model for which the conclusions of Theorem \ref{thm:LNL-FPC} hold is (\ref{CPO model}). Hence our model (\ref{qCPO model}) is the  only one we know for which the conclusions of both theorems hold. 

\vspace{-.1in}
\section{Future work}\label{sec:future}
\vspace{-.1in}
Our presentation of the quantum Fourier transform in $\textbf{qCPO}_{\perp !}$ raises the issue of computational adequacy of our model. We do not have a proof, but we believe a proof strategy similar to that in \cite{LMZ20} will prove this for circuits in our model, but some details still need checking. 

Our model (\ref{qCPO model}) is the only one known that is both sound for Proto-Quipper-M extended with term recursion, and also sound and computationally adequate at non-linear types for LNL-FPC. The latter can be regarded as the circuit-free fragment of Proto-Quipper-M extended with recursive types. Our next goal is to show that (\ref{qCPO model}) is a sound model for Proto-Quipper-M extended with recursive types that is computationally adequate at non-linear types. Since our model also supports affine types, we expect that an adequacy result at all types may follow from results in~\cite{PPRZ19}. 

We are also working on quantizing the probabilistic power domain monad, whose existence would imply (\ref{qCPO model}) supports state preparation. With this in hand, we anticipate extending Proto-Quipper-M with dynamic lifting, i.e., the execution of quantum circuits. We also expect that quantizing the probabilistic power domain will yield another model for the quantum lambda calculus with term recursion \cite{quant-semantics}.

\vspace{-.15in}
\section*{Acknowledgements}                            
\vspace{-.05in}                         
Thanks to Vladimir Zamdzhiev for insightful discussions, and because this work was partly inspired by prior work with him. We also thank Xiaodong Jia for his comments. This work is supported by
  AFOSR under MURI grant FA9550-16-1-0082.


\end{document}